\documentclass[12pt,english]{article}

\usepackage[english]{babel}
\usepackage{multicol}
\usepackage[margin=0.8in]{geometry}
\usepackage{hyperref}
\usepackage[maxbibnames=99]{biblatex}
\usepackage{blindtext}
\hypersetup{
     colorlinks=true,
     linkcolor=blue,
     citecolor=blue,
 }
 \usepackage{breqn}
\usepackage{bm,xfrac,amsmath,amsthm,amssymb,soul, comment, xspace, mathtools}
\usepackage{pst-node}
\usepackage{pst-coil}
\allowdisplaybreaks

\usepackage{cleveref}
\usepackage{tcolorbox}
\usepackage{enumitem}
\usepackage{csquotes}

\usepackage[T1]{fontenc} 
\usepackage{footnote,dsfont}
\usepackage{nicefrac,bm}
\usepackage[ruled,vlined,linesnumbered]{algorithm2e}
\usepackage{caption}
\usepackage{thmtools, float}
\usepackage{thm-restate}
\usepackage{graphicx}
\usepackage{pifont,multirow}
\usepackage{textcomp}
\usepackage{lmodern}
\usepackage{bold-extra}
\usepackage{multicol}
\allowdisplaybreaks
\def\confversion{0} 
\usepackage{ifthen}

\newcommand{\ignore}[1]{}

\ifthenelse{\equal{\confversion}{1}}
{
	\newcommand{\conf}[1]{#1}
}
{
	\newcommand{\conf}[1]{\ignore{#1}}
}
\ifthenelse{\equal{\confversion}{0}}
{
	\newcommand{\full}[1]{#1}
}
{
	\newcommand{\full}[1]{\ignore{#1}}
}

\bibliography{references}

\full{
\title{Cycle Cancellation for Submodular Fractional Allocations and Applications
}

\author{
Chandra Chekuri
\thanks{Supported in part by NSF grant CCF-2402667}\\ \texttt{chekuri@illinois.edu} 
\and
Pooja Kulkarni
\thanks{This work was done partially when the author was a student at UIUC. Supported in part by NSF grant CCF-2334461 and in part by a gift from Adobe to S.Khuller.}\\
\texttt{pooja.kulkarni@northwestern.edu} 
\and
Ruta Mehta
\thanks{Supported in part by NSF grant CCF-2334461.}\\ \texttt{rutameht@illinois.edu}
\and
Jan Vondr\'ak\\
\texttt{jvondrak@stanford.edu}
}
}
\date{}
\conf{
\title{Cycle Cancellation for Submodular Fractional Allocations and Applications}

\author{Anonymous Author(s)}
}

\theoremstyle{plain}
\newtheorem{remark}{Remark}

\newtheorem{lemma}{Lemma}
\newtheorem*{lemma*}{Lemma}

\newtheorem{theorem}{Theorem}
\newtheorem{corollary}{Corollary}

\newtheorem{claim}{Claim}

\newtheorem{definition}{Definition}

\newcommand{\A}{\mathcal{A}}
\newcommand{\N}{{\mathcal N}}

\newcommand{\M}{{\mathcal M}}
\newcommand{\Vals}{(v_i)_{i\in\N}}

\newcommand{\derivative}[2]{ V'_{#1,#2} }
\newcommand{\truncate}[1]{#1_{\downarrow i,j}}

\newcommand{\MMSinssym}{(\N,\M,\Vals)}

\newcommand{\vecx}{\mathbf{x}}

\newcommand{\classP}{{\sf P }}
\newcommand{\classNP}{{\sf NP }}
\newcommand{\classAPX}{{{\sf APX }}}

\newcommand{\SPLC}{{\sf SPLC }}

\newcommand{\MMS}{{\sf{MMS}}\xspace}

\newcommand{\NSW}{{\sf NSW} \xspace}
\newcommand{\OPT}{{\sf OPT} \xspace}
\newcommand{\ALG}{{\sf ALG} \xspace}

\newcommand{\R}{\mathcal R}

\newcommand{\eps}{{\epsilon}}

\newcommand{\bx}{{\mathbf x}}
\newcommand{\by}{{\mathbf y}}
\newcommand{\bz}{{\mathbf z}}
\newcommand{\be}{{\mathbf e}}
\newcommand{\bone}{{\mathbf 1}}

\let\oldnl\nl
\newcommand{\nonl}{\renewcommand{\nl}{\let\nl\oldnl}}

\DeclareMathOperator*{\argmax}{\arg\!\max}

\renewcommand{\nonl}{\renewcommand{\nl}{\let\nl\oldnl}}
\long\def\symbolfootnote[#1]#2{\begingroup%
\def\thefootnote{\fnsymbol{footnote}}\footnote[#1]{#2}\endgroup}

\begin{document}

\maketitle

\begin{abstract}
We consider discrete allocation problems where $m$ indivisible goods need to be allocated among $n$ agents. When agents' valuation functions are additive, the well-known {cycle canceling lemma} \cite{LenstraST90,ShmoysT93,PST95} plays a key role in the design and analysis of {rounding} algorithms. 
In this paper, we prove an analogous lemma for the case of submodular valuations. Our algorithm removes cycles in the support graph of a fractional allocation while guaranteeing that each agent’s value, measured using the multilinear extension, does not decrease. 

We demonstrate applications of the cycle-canceling algorithm, along with other ideas, to obtain new algorithms and results for three well-studied allocation objectives: max-min (Santa Claus problem), Nash social welfare ($\NSW$), and maximin-share ($\MMS$). For the submodular $\NSW$ problem, we obtain a $\frac{1}{5}$-approximation; for the $\MMS$ problem, we obtain a $\frac{1}{2}(1-1/e)$-approximation through new simple algorithms. For various special cases where the goods are ``small'' valued or the number of agents is constant, we obtain tight/best-known approximation algorithms. All our results are in the value-oracle model. 
\end{abstract}

\section{Introduction}

Allocation problems are fundamental to both theory and practice. They arise naturally in diverse applications  from optimization to fair division, and have been an influential source of techniques in discrete (combinatorial) optimization. In this paper we give a new tool for addressing allocation problems with submodular valuations.

Consider the problem of allocating a set $\M$ of $m$ \emph{indivisible} items to a set $\N$ of $n$ agents with diverse preferences. An allocation can be viewed as a partition of $\M$ into sets $S_1, S_2,\ldots,S_n$ with $S_i$ being the set given to agent $i$.
Each agent $i$ has a valuation function $v_i : 2^{\M} \rightarrow \mathbb{R}_+$ where $v_i(S)$ is the value that they obtain for subset $S \subseteq \M$. We will assume $v_i$s to be {\em monotone  submodular}\footnote{A real-valued set function $f:2^\M \rightarrow \mathbb{R}_+$
is submodular if $f(A+g) -f(A) \ge f(B+g) - f(B)$ for any $A \subseteq B$ and $g \in B \setminus A$; in other words the function exhibits diminishing marginal
utility. Monotonicity means that $f(A) \le f(B)$ for all $A \subseteq B$.}, capturing decreasing marginal returns, and normalized so that $v_i(\emptyset)=0$. Throughout this work, we denote an instance of the allocation problem by $\MMSinssym$ and the set of all allocations as $\Pi_{|\N|}(\M)$. 

Allocation problems with submodular valuations are extensively studied within optimization \cite{mirrokni2008tight, chekuri2010dependent, dobzinski2006improved} 
and fair division \cite{lipton2004approximately, ghodsi2018fair, UziahuF23, GKKsubmodnsw, garg2023approximating} communities under various objective functions. Also, see Section \ref{sec:applns} for more detailed related works. A prominent example is to allocate items to maximize social welfare: find an allocation $S_1,\ldots,S_n$ to $\max \sum_{i \in \N} v_i(S_i)$.
While this is trivial for additive valuations (simply assign each good to the agent that values it the most), it is an interesting optimization problem for
submodular valuations, and is $\classNP$-Hard. A simple greedy algorithm achieves an approximation ratio of $1/2$ \cite{nemhauser1978analysis, lehmann2001combinatorial} while a much more sophisticated approach achieves a tight approximation ratio of $(1-1/e)$ \cite{vondrak2008optimal, calinescu2011maximizing}. The latter result was instrumental in developing the multilinear relaxation approach for submodular maximization problems. In this paper, we are concerned with more challenging objectives arising from fairness, efficiency, and other considerations, with the goal of designing unifying tools and techniques. 
Three such objectives are: (i) Max-min allocation (also referred to as the (Submodular) Santa Claus problem) \cite{lipton2004approximately}, (ii) Nash Social Welfare ($\NSW$) \cite{kaneko1979nash}, and (iii) Maximin-share ($\MMS$) \cite{budish2011combinatorial}. These problems are well-known to be challenging even when agents' valuation functions are additive, i.e, $v_i(S) = \sum_{j \in S} v_i(\{j\})$ for all $i \in \N$. See Section \ref{sec:applns} for more details about these problems.

Relax-and-round is a natural approach for most of these problems, where first a good {\em fractional allocation} is obtained and then rounded to an integral allocation while incurring some loss. The purpose of this paper is to develop a widely applicable technical tool that helps round a fractional allocation effectively under submodular valuations. We further highlight the utility of this tool by instantiating it on the above mentioned three objectives. 
Our tool is inspired by a corresponding result in the additive setting: A fractional allocation is an $m\cdot n$ dimensional vector $\bx$ which we interpret as consisting of values $x_{i,j}$ for $i \in \N$ and $j \in \M$; $x_{i,j} \in [0,1]$ is the fractional amount of item $j$ that is assigned to agent $i$. We require that $\sum_{i} x_{i,j} = 1$ for each item $j$. We let $\bx_i$ denote the $m$-dimensional vector consisting of the variables $x_{i,j}$, $j \in [m]$. Each fractional allocation $\bx$ defines an allocation graph $G_{\bx}$ which is a bipartite graph with agents on one side and items on the other side and an edge between $i$ and $j$ labeled with the number $x_{i,j}$ iff $x_{i,j} > 0$; in other words it is a representation of the support of $\bx$. For a fractional allocation $\bx$ the value of agent $i$ 
in the additive setting, denoted by $v_i(\bx_i)$ is naturally defined as $\sum_{j} v_i(\{j\}) x_{i,j}$. Then the following rounding procedure is known.

\begin{lemma} [\cite{LenstraST90,ShmoysT93}]
    \label{lem:cycle-cancel}
Let $\bx$ be a fractional allocation and $v_i$ additive functions. Then there is another allocation $\by$ in the support of $\bx$ such that (i) the allocation graph $G_{\by}$ has no cycles and (ii) for all agents $i$, $v_i(\by_i) \ge v_i(\bx_i)$. Furthermore, the allocation $\by$ can be rounded to an integral allocation $\A=(A_1,\ldots, A_n)$ such that
  for all agents $i$, $v_i(A_i) \ge V_i(\by_i) - \alpha_i$ where $\alpha_i = \max_{j: y_{i,j} > 0} v_i(\{j\})$.
\end{lemma}

\begin{corollary}
  Suppose $\bx$ is a fractional allocation such that $\max_{j: x_{i,j} > 0} v_{i}(\{j\}) \le \eps v_i(\bx_i)$ for each agent $i$. Then there is an integral allocation $\bz$ such that $v_i(\bz_i) \ge (1-\eps)v_i(\bx_i)$.
\end{corollary}

The corollary immediately allows one to find a good approximation in the so-called ``small'' items scenario. One can view the lemma and the corollary as interpolating between the fully fractional setting (which can be solved optimally via LP) and the discrete setting in which no single item is too important to an agent. The utility of the lemma can be understood by considering an alternative rounding strategy: randomly allocate each item independently to an agent based on fractional solution. The parameters obtained by such a rounding are much weaker --- a logarithmic factor in the number of agents is typically lost since one has to rely on concentration inequalities followed by a union bound.

Here we generalize the preceding lemma to the
submodular setting. In order to do this we need to work with a continuous extension of submodular functions. As mentioned above, for maximization problems the multilinear extension \cite{calinescu2007maximizing}
has played a natural role.

\begin{definition}[Multilinear Extension of $f$]
    \begin{equation}
    {F}(x) \coloneqq \sum_{S \subseteq V} f(S) \prod_{j \in S} x_j \prod_{j \in V \setminus S} (1 - x_j)
    \end{equation}
\end{definition}
It inherits some useful properties of submodularity and moreover, its value and partial derivatives can be evaluated for a given point $\bx$ via random sampling.

Our first contribution is to formulate and prove the following theorem.
\begin{restatable}{theorem}{thmmain}\label{thm:main}
    Given an instance of the allocation problem $\MMSinssym$ with submodular valuations $v_i$, and a fractional allocation $\vecx$ of the goods to the agents, there is a fractional allocation $\vecx'$ such that (i) the allocation graph $G_{\vecx'}$ is acyclic, and (ii) $V_i(\vecx') \geq V_i(\vecx)$ for all agents $i \in \N$; where $V_i$ is the multilinear extension of $v_i$.
\end{restatable}

We state the preceding theorem in a non-constructive way. A polynomial-time algorithm requires one to use approximations due to the sampling aspects of the multilinear extension. We defer formal details of this to a full version of the paper since they are somewhat standard. An efficient version of the lemma yields an allocation
that loses a $(1-o(1))$-factor in the value for each agent which we ignore in this version. See the end of Section \ref{sec:rounding} for more details. 

Next we discuss application of this lemma to obtained improved algorithms and/or results for the three objectives mentioned above.

\subsection{Applications}\label{sec:applns}
We demonstrate the utility of Theorem \ref{thm:main} on three problems that arise in fair allocation. The focus is not on obtaining improved approximation factors, but rather to illustrate how the cycle canceling procedure can be applied in conjunction with other ideas to design new relaxation-based algorithms. Table \ref{tab:summary-results} gives a summary of our results.

\begin{table}[t]
    \centering
    \begin{tabular}{|c|c|c|c|c|}
    \hline
         Fairness Notion & Problem Setting & Lower Bound & Prior Work & This Paper \\
         \hline \hline
         Max-min & $v_i(j) \leq \epsilon \OPT$ $\forall i, j$ & -- & -- & $(1-\sfrac{1}{e} - \epsilon)$\\
         \hline
         Max-min & Constant $n$ & $(1-\sfrac{1}{e} + \epsilon)$ \cite{mirrokni2008tight} **& $(1-\sfrac{1}{e} - \epsilon)$ \cite{chekuri2010dependent}& $(1-\sfrac{1}{e} - \epsilon)$* \\
         \hline
         $\NSW$ & General & $(1-\sfrac{1}{e} + \epsilon)$ \cite{GKKsubmodnsw} & $\approx 1/3.9$ \cite{bei2025nash} & $1/5$ \\
         \hline
         $\NSW$ & Constant $n$ & $(1-\sfrac{1}{e} + \epsilon)$ \cite{GKKsubmodnsw}&  $(1-\sfrac{1}{e} -\epsilon)$ \cite{GKKsubmodnsw} & $(1-1/e-\epsilon)$* \\
         \hline
         $\MMS$ & General   & $(1-1/e + \epsilon)$ \cite{mirrokni2008tight}\footnotemark** & $\approx 0.370$ \cite{UziahuF23} & $ \approx 0.316$\\
         \hline
          $\MMS$ & $v_i(j) \leq  \epsilon \MMS_i$ $\forall i, j$ & --  & $\approx 0.370$ \cite{UziahuF23}& $\left( 1 - \sfrac{1}{e} -\epsilon \right)$ \\
          \hline
    \end{tabular}
    \caption{Summary of Applications of Theorem \ref{thm:main}. For the result marked with (*), while our result matches the existing result, we need $\epsilon$ to be a constant in this paper whereas the prior work required $\epsilon$ to be $\Omega(\sfrac{1}{\log n})$. Results marked with (**) are information theoretic lower bounds whereas the other results are $\classNP$-hardness.}
    \label{tab:summary-results}
\end{table}
 \footnotetext{The results from \cite{mirrokni2008tight} are for welfare maximization. However, the structure of instance used for the reduction directly implies the same bound for the problems mentioned here.}

All the three applications we discuss in this paper have seen a lot of work in the literature. While the table gives the current best-known guarantees, we discuss more details here. 

\subsubsection{Santa Claus (Max-min)}\label{sec:max-min-intro}  In Section \ref{sec:max-min}, we investigate the max-min objective where given an instance of the fair allocation problem, $\MMSinssym$, the goal is to output an allocation $\A = \{A_1, \ldots, A_n\} \in \Pi_{|\N|}(\M)$ that maximizes $\min_{i \in \N} v_i(A_i)$. It is alternatively referred to as the Santa Claus problem.

\paragraph{Related work.} Max-min was first considered as a fair allocation problem by \cite{lipton2004approximately}. Lot of work has since been done on this problem with the approximability not well understood even in the additive setting. \cite{bezakova2005allocating} showed that one cannot approximate max-min to a factor better than $2$ in the additive setting unless $\classP=\classNP$. On the positive side, \cite{bansal2006santa} gave an $O(\frac{\log \log m}{\log \log \log m})$-approximation algorithm. In the general setting, \cite{ChakrabartyCK09} gave a $1/n^{\sfrac{1}{\epsilon}}$ approximation algorithm in running time $n^{O(\sfrac{1}{\epsilon})}$. In the restricted assignment case (i.e., for each agent $v_i(\{j\}) \in \{v_j,0\}$), \cite{cheng2019restricted} gave a $(4+\delta)$-approximation in $O(\text{poly}(m,n)n^{O(\sfrac{1}{\delta})})$ time. Both the results were extended to the submodular setting:  \cite{BamasMR25} gave a $\Omega(\frac{1}{{poly}(\log n)})$-approximation algorithm with running time $n^{O(\sfrac{\log n}{\log \log n})}$ and for any fixed $\epsilon > 0$; a $\frac{1}{n^{\sfrac{1}{\epsilon}}}$-approximation algorithm with running time $n^{O(\sfrac{1}{\epsilon})}$. For the restricted case with submodular valuations, \cite{BamasGR21} gave an $\Omega(\frac{1}{\log \log n})$-approximation algorithm. Here, restricted submodular Santa Claus is defined as follows: there is a common submodular valuation function $v$, each agent has a subset $X_i \subseteq \M$ and the value of agent $i$ is $v_i(S) = v(\M \cap X_i)$. Another result known for additive valuations was by \cite{bezakova2005allocating} who showed that there is always an allocation which gives each agent a value of at least $(\OPT-\max_{i \in \N, j \in \M} v_i(j))$ where $\OPT$ is the optimal max-min value.

\paragraph{Our Contributions.} In this work, we study the {\em submodular} Santa Claus problem. We show an additive approximation guarantee for it analogous to the result by \cite{bezakova2005allocating} for additive valuations. While additive and multiplicative approximation guarantees are incomparable in general, when the goods are small valued, additive approximation imply multiplicative approximations. We thereby achieve the first non-trivial multiplicative approximation for the small-goods case with submodular valuations. Furthermore, we give a $(1-1/e-\epsilon)$ approximation for the case when the instance has constantly many agents in a more efficient manner than was previously achieved by \cite{chekuri2010dependent} (see Remark \ref{rmk:max-min}).

\subsubsection{Nash social welfare ($\NSW$)}\label{sec:nsw-intro} In Section \ref{sec:nsw}, we investigate Nash social welfare which is a popular fairness \emph{and} efficiency notion. The objective is to find an allocation $\A=(A_1, \ldots, A_n)$ that maximizes the geometric mean of agent valuations i.e., \( \A = \argmax_{S \in \Pi_{|\N|}(\M)}\allowbreak \left(\prod_{i \in [n]} v_i(S_i)\right)^{\sfrac{1}{n}}.\)

\paragraph{Related work.} The Nash social welfare ($\NSW$) objective has been proved to offer a natural balance between utilitarian welfare (sum of utilities) and fairness (max--min welfare)~\cite{kaneko1979nash,caragiannis2019unreasonable}. Even for additive valuations, maximizing $\NSW$ is $\classAPX$-hard \cite{Lee17} and obtaining a constant factor approximation was challenging. The first constant factor for additive valuations was presented by Cole and Gkatzelis \cite{cole2018approximating}, using a market equilibrium approach. Alternative approaches were found afterwards, using either market equilibria or fractional relaxations, that provide constant factors for additive, budget-additive and piecewise-linear concave valuations \cite{anari2016nash,anari2018concave,cole2017convex,barman2018finding,garg2018approximating}. 
For submodular valuations,~\cite{GKKsubmodnsw} gave the first approximation independent of the number of goods using a ``matching-rematching'' along with greedy allocation approach. \cite{garg2021rado} extended the matching-rematching approach along with fractional allocations and proved a constant factor for ``Rado valuations'' (a subclass of submodular valuations related to matroids). This framework was thereafter extended to a $1/380$-approximation for submodular valuations by~\cite{li2022constant}, which was subsequently improved to a $1/4$-approximation using matching-rematching and local search \cite{garg2023approximating}. A recent line of work~\cite{feng2025note,feng2025constant,bei2025nash} studied the \emph{weighted} $\NSW$ problem, maximizing the weighted geometric mean of valuations, and obtained constant-factor approximations for additive and submodular valuations even in this setting. For additive valuations, their guarantees match the best known bound for unweighted $\NSW$ ($e^{1/e}$) from~\cite{barman2018finding}; for submodular valuations, the most recent approximation factor is $\sim1/5.18$~\cite{bei2025nash} via computer-aided analysis. These algorithms are based on a configuration-type relaxation and randomized rounding. 

\paragraph{Our Contributions.} Our primary result here is to present a simple algorithm that still achieves a small constant-factor approximation for submodular $\NSW$. Our algorithm is based on the matching-relaxation-rematching framework like \cite{GKKsubmodnsw, garg2021rado, li2022constant} and provides a $1/5$-approximation guarantee. Additionally, we also give a $(1-1/e-\epsilon)$-approximation for the case when the instance has constantly many agents in a more efficient manner than was previously done by \cite{GKKsubmodnsw} (see Remark \ref{rmk:nsw}).

\subsubsection{Maximin Share ($\MMS$)}\label{sec:mms-intro} In Section \ref{sec:mms}, we look at the Maximin share objective. An agent's $\MMS$ value, denoted by $\MMS_i$ is the maximum value she can ensure for herself if she divides the goods into $n$ parts and chooses the worst part for herself. 

\paragraph{Related work.} Maximin Share $(\MMS)$ was defined as a fairness notion by \cite{budish2011combinatorial}. \cite{procaccia2014fair} proved that $\MMS$-allocations i.e., allocations that give every agent their $\MMS$ value need not exist even with additive valuations. Therefore, a series of works that obtain $\alpha$-$\MMS$ allocation, i.e., each agent receives a value of $\alpha$-$\MMS$ have been studied for additive valuations \cite{procaccia2014fair, amanatidis2017approximation, garg2019approximating, garg2020improved, akrami2024improving, akrami2024breaking} culminating in a $3/4 + O(1)$-approximation. The algorithms in these works are mostly combinatorial and rely on either matching-based or greedy-based techniques. For general valuations, \cite{barman2020approximation, ghodsi2018fair} initiated the study for approximate $\MMS$ allocations. \cite{barman2020approximation} gave a $\sim 1/10$ approximation algorithm via round-robin allocation and \cite{ghodsi2018fair} gave a $1/3$-approximation via local search. Recently, \cite{UziahuF23} gave an algorithm that provides $10/27$-approximation to $\MMS$ using a greedy algorithm. For a subclass of submodular valuations called $\SPLC$, \cite{chekuri20241} gave a $1/2$ approximation algorithm based on market-equilibrium solution followed by rounding. Several other works on $\MMS$ on other valuation functions or in constrained settings exist that we do not list here.
\paragraph{Our Contributions.} Our main theorem naturally extends the equilibrium-rounding approach to submodular valuations to give a $\frac{1}{2}(1-1/e-o(1))$-approximation algorithm. Additionally, when the goods are small compared to the $\MMS$ value (at most $\epsilon \MMS$), we get a novel $(1-1/e-\epsilon)$-$\MMS$ allocation using the same approach.

\section{Cycle-Cancellation under Multilinear Extension}
\label{sec:rounding}
In this section we present the main technical contribution of this work -- Theorem \ref{thm:main}. This theorem states that given an instance $\MMSinssym$ of the allocation problem and any fractional allocation $\vecx$, we can convert it to an allocation with at most $|\N|-1$ fractional variables whose allocation graph forms a forest, while retaining the value received by any agent \emph{under the multilinear relaxation}. It is well-known that when the function is submodular, the multilinear extension is convex in any $\be_i - \be_j$ direction \cite{calinescu2007maximizing}. We first generalize this to show that the multilinear extension is convex in any direction given by $t_i \be_i - t_j \be_j$, where $t_i, t_j \geq 0$. This property has been recently observed independently in \cite{lijun2025fair}. This allows us to argue how the function value changes when we move goods along a cycle allowing us to eventually cancel cycles in the allocation graph.

\thmmain*

We remark that the theorem applies to all submodular functions (possibly non-monotone). 
Before we prove the theorem, we see what happens to the value of the multilinear extension when for two goods $i$, $j$ with $i \neq j$, we increase the allocation of good $i$ by $t_i$ and decrease the amount of good $j$ by $t_j$. We remark that an equivalent lemma appears in a recent work on the knapsack problem with a submodular objective \cite{lijun2025fair}.

\begin{lemma}\label{lem:single-function}
   Given a submodular function $f$ with multilinear extension $F$, and allocations $\vecx = (x_1, \ldots, x_m)$ and $\Bar{\vecx} = (x_1, \ldots, x_i + t_i, \ldots, x_j - t_j, \ldots, x_m)$, for any $t_i, t_j \geq 0$, we have
   \begin{equation*}
       F(\Bar{\vecx}) - F({\vecx}) \geq t_i \frac{\partial F}{\partial x_i}\Big|_{\bx} - t_j \frac{\partial F}{\partial x_j}\Big|_{\bx}. 
   \end{equation*}
\end{lemma}

\begin{proof}
    Let $\vecx = (x_1, \ldots, x_m)$, $\Bar{\vecx} = (x_1, \ldots, x_i + t_i, \ldots, x_j - t_j, \ldots, x_m)$ and $P_\bx(R) = \prod_{g \in R} x_g$ $\prod_{g \in [m] \setminus R} (1-x_g)$. Note that $\vecx$ and $\Bar{\vecx}$ differ in probabilities only on the goods $i$ and $j$. Therefore, when we expand the multilinear extension, the only places the terms differ are on the probabilities of selecting (or not selecting) $i$ or $j$. We explicitly separate these terms. Towards this, we denote $``R:i,j"$ to indicate all sets that include the elements $i$ and $j$, $``R:\tilde{i}, \tilde{j}"$ are all sets that do not include the elements $i$ and $j$. $``R:\tilde{i}, j"$ and $``R:i,\tilde{j}"$ are defined analogously. Further, we define $\truncate{P}(R) = \prod_{g \in R\setminus \{i,j\} }x_g \prod_{g \in [m] \setminus R \setminus \{i,j\} }(1-x_g)$ i.e., the probability of picking a set $R$ according to a distribution $\vecx$ except that we ignore the probabilities corresponding to $i$ and $j$. With this notation, we can write the difference of multilinear extensions as:
    \begin{align*}
        F(\Bar{\vecx}) - F(\vecx) = & \sum_{R} (P_{\bar{\bx}}(R) - P_{\bx}(R))  f(R) \\
        = & \sum_{R : \tilde i, \tilde j} \truncate{P}(R) (t_j (1 - x_i) - t_i(1 - x_j) - t_i t_j) f(R) \\&+ \sum_{R : i, \tilde j} \truncate{P}(R) (t_i t_j + t_i(1 - x_j) + x_it_j)f(R) \\&+ \sum_{R : \tilde i, j} \truncate{P}(R) (t_it_j - (1-x_i)t_j - t_ix_j)f(R) \\
        &+ \sum_{R : i, j} \truncate{P}(R) (t_ix_j - x_it_j - t_it_j) f(R).
    \end{align*}
    We rewrite the above equations by combining the terms in $t_it_j$, $t_i$ and $t_j$ together and we get the following equation
    \begin{align*}
        F(\Bar{\vecx}) - F(\vecx) = &t_it_j \sum_{R : \tilde i, \tilde j} \truncate{P}(R) (-f(R) + f(R+i) + f(R+j) - f(R+i+j)) \\
        &+ t_i \sum_{R : \tilde i, \tilde j} \truncate{P}(R) (-(1-x_j)f(R) + (1- x_j)f(R+i) - x_j f(R+j) +x_jf(R+i+j))  \\
        &+ t_j \sum_{R: \tilde i, \tilde j} \truncate{P}(R)((1-x_i)f(R) - x_i f(R+i) - x_j f(R+j) -x_i f(R+i+j))
    \end{align*}
    Let $ f(R+i) + f(R+j) - f(R) - f(R+i+j) = \alpha$. Note that the coefficient of $t_i$ is $\frac{\partial F}{\partial x_i}\Big|_{\bx}$ and that of $t_j$ is $- \frac{\partial F}{ \partial x_j}\Big|_{\bx}$. Substituting these values, we get
    \begin{align*}
        F(\Bar{\vecx}) - F(\vecx) &= \alpha t_it_j  + t_i \frac{\partial F}{\partial x_i}\Big|_{\bx} - t_j \frac{\partial F}{\partial x_j}\Big|_{\bx} \geq t_i \frac{\partial F}{\partial x_i}\Big|_{\bx} - t_j \frac{\partial F}{\partial x_j}\Big|_{\bx}
    \end{align*} 
    where the second inequality follows since $\alpha \geq 0$ for submodular valuation functions and we assumed $t_i$, $t_j$ to be non-negative.
\end{proof}

We can now prove Theorem \ref{thm:main}. Recall that the instance of allocation problem is denoted by $\MMSinssym$ and we assume that each $v_i$ is submodular. We denote the multilinear extension of $v_i$ by $V_i$. Note that while we write $V_i(\vecx)$, the value of $V_i$ only depends on the part of $\bx$ corresponding to agent $i$'s allocation $\bx_i$. Additionally, we use $\derivative{i}{j}$ to denote the partial derivative of $V_i$ at $\vecx_i$  with respect to good $j$. 

\begin{proof}[Proof of Theorem \ref{thm:main}.] The main idea is that if there is a cycle with fractional values $x_{ij}$, we can move values along the cycle consistently so that no agent loses in value, and some variable becomes an integer. For any given allocation $\vecx$, consider a cycle  $a_0 - g_0 - a_1 - g_1 - \cdots - a_{\ell-1} - g_{\ell-1} - a_0$ where $a_i$ are the agents and $g_i$ are the goods. To reallocate along the cycle, we need to maintain two constraints: (1) the fractional allocation of each good remains constant, (2) no agent's value (under the multilinear extension) decreases. To maintain the first constraint, we will change the allocation of good $i$ by $\delta_i$ on one edge and $-\delta_i$ on the other. Therefore, a reallocation strategy looks as follows:
\begin{align*}
    x'_{a_ig_i} &= x_{a_ig_i} + \delta_i \\
    x'_{a_{i+1}g_i } &= x_{a_{i+1} g_i} - \delta_i
\end{align*}
with index operations modulo $\ell$.
By Lemma \ref{lem:single-function}, for agent $a_i$ not to lose value, it is sufficient to satisfy:
\begin{align}\label{eqn:constrait-agents-value}
     \delta_i \derivative{a_i}{g_i} - \delta_{i-1} \derivative{a_i}{g_{i-1}} \geq 0
\end{align}
where, the index operations are modulo $\ell$. We get $\ell$ such inequalities corresponding to the $\ell$ agents on the cycle. We claim that there is always a non-trivial solution to these inequalities, for any choice of $\derivative{a_i}{g_i}$, $\derivative{a_i}{g_{i-1}}$. 

First, if $\derivative{a_{i}}{g_i} = 0$ for some agent $i$, we can set $\delta_{i'} = 0$ for all $i' \neq i$, and choose $\delta_i \neq 0$ in order to make the gain of agent $i+1$ positive. Note that the gain of every other agent is $0$. Similarly, if $\derivative{a_i}{g_{i-1}} = 0$, we can set $\delta_{i'} = 0$ for all $i' \neq i-1$ and choose $\delta_{i-1} \neq 0$ in order to make the gain of agent $i-1$ positive. Again, the gain of all other agents is $0$ in this case.

Hence, we can assume that $\derivative{a_i}{g_i} \neq 0$ and $\derivative{a_i}{g_{i-1}} \neq 0$ for all agents on the cycle. In this case, for any value $\delta_0$, we can iteratively find values of $\delta_2, \ldots, \delta_{\ell-1}$ such that $\delta_i \derivative{a_i}{g_i} - \delta_{i-1} \derivative{a_i}{g_{i-1}} = 0$ for $1 \leq i \leq \ell-1$. These conditions mandate that $\delta_{i} = \delta_{i-1} \frac{\derivative{a_i}{g_{i-1}}}{\derivative{a_i}{g_{i}}}$, so we obtain inductively
\begin{align*}
    \delta_{\ell-1} = \delta_{0} \frac{\derivative{a_1}{g_0}}{\derivative{a_1}{g_1}} \frac{\derivative{a_2}{g_1}}{\derivative{a_2}{g_2}} \cdots \frac{\derivative{a_{\ell-1}}{g_{\ell-2}}}{\derivative{a_{\ell-1}}{g_{\ell-1}}}.
\end{align*}
The last remaining condition is that
\begin{align*}
    \delta_0 \derivative{a_0}{g_0} - \delta_{\ell-1} \derivative{a_0}{g_{\ell-1}} \geq 0
\end{align*}
which is equivalent to
\begin{align*}
    \delta_0 \derivative{a_0}{g_0} - \delta_0 \frac{\derivative{a_1}{g_0}}{\derivative{a_1}{g_1}} \frac{\derivative{a_2}{g_1}}{\derivative{a_2}{g_2}} \cdots \frac{\derivative{a_{\ell-1}}{g_{\ell-2}}}{\derivative{a_{\ell-1}}{g_{\ell-1}}} \derivative{a_0}{g_{\ell-1}} \geq 0.
\end{align*}
We still have a free choice of $\delta_0 \neq 0$, so we can choose its sign appropriately so that the last condition is satisfied. By construction, this satisfies all the conditions (\ref{eqn:constrait-agents-value}).

We want to choose the magnitude of $\delta_0$ so that the modified solution $\bx'$ is feasible and at least one variable on the cycle becomes an integer. This can be accomplished by choosing the supremum of $|\delta_0|$ such that $\bx'$ defined by the changes described above is still feasible. We claim that such $\bx'$ must have a new integer variable: If all the variables $x'_{ij}$ on the cycle are still fractional, then there is a larger value of $|\delta_0|$ that we can choose and the modification is still feasible. Finally, the supremum of all feasible choices of $|\delta_0|$ also defines a feasible solution, because the assignment polytope is a closed set.
\end{proof}

\noindent
{\em Notes on computational aspects of the procedure:}
In the discussion above, we ignore the computational issues of computing or estimating the values of $V_i(\bx)$ and $\frac{\partial V_i}{\partial x_{ij}}$. These issues can be handled by techniques similar to prior work using the multilinear extension: The values of $V_i(\bx)$ and its partial derivatives are expectations over a certain product distribution and hence can be estimated by random sampling. The sampling error can be made smaller than $\frac{\alpha_i}{poly(m,n)}$ where $\alpha_i = \max_{j \in \M: x_{ij}>0} v_i(\{j\})$. Since variables are changing by at most $1$ in each step, the true change in $V_i(\bx)$ might be off by $\frac{\alpha_i}{poly(m,n)}$ compared to our estimates. The number of rounding steps is polynomial in $m$ and $n$ (at least one new variable becomes an integer in each step), and hence we can choose the sampling parameters so that the loss for each agent is less than $\frac{\alpha_i}{poly(m,n)}$ at the end. In many applications, this additive error is negligible compared to the loss of $\alpha_i$ which we shall incur in any case due to the rounding on the remaining forest. For completeness, we give our computational result:
\begin{theorem}\label{thm:main-comp}
    Given an instance of the allocation problem $\MMSinssym$ with submodular valuations $v_i$, and a fractional allocation $\vecx$ of the goods to the agents, in randomized polynomial time, we can compute with high probability another fractional allocation $\by$ such that the support of $\by$ is acyclic, and $V_i(\by) \geq V_i(\vecx) - \frac{\alpha_i}{poly(m,n)}$ for all agents $i \in [n]$ where $V_i$ is the multilinear extension of $v_i$ and $\alpha_i = \max_{j \in \M: x_{ij}>0} v_i(\{j\})$.
\end{theorem}

\section{Applications}
In this section we discuss applications of the cycle-cancellation. All applications we discuss use the same rounding based on cycle-cancellation which we call \textsf{NonUniformPipageRounding} and is outlined in Algorithm \ref{alg:non-uniform-pipage}. Similar rounding has appeared in the context where valuations are additive or $\SPLC$\footnote{$\SPLC$ is a class of functions that capture additive valuations and are a subset of submodular valuations.} (eg., \cite{cole2018approximating, chekuri20241}).
\begin{algorithm}[h!]
    \caption{{\sf NonUniformPipageRounding}$(\MMSinssym, \vecx)$}
    \label{alg:non-uniform-pipage}
    \SetKwInOut{Input}{Input}\SetKwInOut{Output}{Output}
    Create $\vecx'$ from $\vecx$ using the cycle-cancellation process of Theorem \ref{thm:main}\\
    In the forest of $G_{\vecx'}$, root each tree at an arbitrary agent \label{step:root} \\
    Allocate each fractional good to its parent in the rooted tree
\end{algorithm}

We can prove the following guarantee of this algorithm.
\begin{lemma}\label{lem:nonunif-pipage}
    Given an instance of the allocation problem, $\MMSinssym$ and any fractional allocation $\vecx$, in randomized polynomial time, Algorithm \ref{alg:non-uniform-pipage} \emph{(\textsf{NonUniformPipageRounding})} finds an allocation $\A=(A_1, \ldots, A_{|\N|}) \in \Pi_{|\N|}(\M)$ such that $v_i(A_i) \geq V_i(\vecx) - v_i(\ell_i)$ for all $i \in \N$.
\end{lemma}
\begin{proof}
    From Theorem \ref{thm:main}, the agents lose no value in converting to $\vecx'$. Then, since the graph is acyclic, step \ref{step:root} can be executed. Since each agent in the rounding can lost at most one good (the parent good), the agents lose a single good. By submodularity (in fact, subadditivity), the lemma follows.
\end{proof}
\begin{remark}
    As mentioned at the end of previous section, we ignore the computational aspects of cycle-cancellation procedure. All our applications use the $\sf{NonUniformPipageRounding}$ where the rounding incurs loss of one good. The loss due to computational aspects of cycle-cancellation (result of Theorem \ref{thm:main-comp}) is negligible compared to this loss and therefore can be safely ignored. For keeping the notation clean, we do not add that loss here.
\end{remark}
\subsection{Approximating Submodular Max-min }\label{sec:max-min}
The max-min problem also known as the Santa Claus problem is defined as: given an instance of the allocation problem, $\MMSinssym$, output an allocation $\A = (A_i, \ldots, A_{|\N|}) \in \Pi_{|\N|}(\M)$ to maximize $\min_{i \in \N} v_i(A_i)$. This value is called the max-min value and we denote it as $\OPT = \max_{\A \in \Pi_{|\N|}(\M)} \min_{i \in \N} v_i(A_i)$.

As discussed in Section \ref{sec:max-min-intro}, our primary aim is to provide an additive approximation for the Submodular Santa Claus problem. We show that this directly gives a multiplicative approximation for ``small'' goods. Additionally, we also show how to use cycle-cancellation to achieve a multiplicative approximation for the special case when the instance has a fixed number of agents. The key to these results is the \textsf{NonUniformPipageRounding} subroutine (Algorithm \ref{alg:non-uniform-pipage}). 
To use the subroutine, we naturally need a fractional solution. To obtain this fractional solution for the max-min problem, we use the following result by \cite{chekuri2010dependent,chekuri2015multiplicative}.

\begin{theorem}[\cite{chekuri2010dependent,chekuri2015multiplicative}]\label{thm:cont-greedy}
  Let $f_1, \ldots, f_n$ be monotone non-negative submodular functions over a common ground set $N$, $B_1, B_2, \ldots, B_n \in \mathbb{R}_+$ and let $Q \in [0,1]^{|N|}$ be a solvable\footnote{A polytope is called solvable if we can optimize linear functions over the polytope in polynomial time.} polytope. For every fixed $\epsilon > 0$ there is a polynomial time randomized algorithm that either outputs a point $\bx \in Q$ such that for all $i$, \[ F_i(\bx) \geq (1-1/e-\epsilon) B_i \] or correctly decides that there is no point $\bx$ in $Q$ such that $F_i(\bx) \geq B_i$ for all $i$.
\end{theorem}

\begin{theorem}\label{thm:add-santa-claus-result}
    Consider an instance, $\MMSinssym$ of the Submodular Santa Claus problem with the optimal value $\OPT$. Denote the maximum value of any good as $\emph{Max} = \max_{i \in \N, j \in \N} v_i(j)$. For any fixed $\epsilon > 0$, one can obtain in randomized polynomial time an allocation $\A = (A_1, \ldots, A_{|\N|})$ such that $v_i(A_i) \geq (1-1/e -\epsilon)\OPT - \emph{Max}$.
\end{theorem}
\begin{proof}
    First, we obtain a fractional allocation $\vecx = (\vecx_1, \ldots, \vecx_{|\N|})$ such that $V_i(\vecx_i) \geq (1-1/e-\epsilon) \OPT$. This can be done as follows.  We guess the value of $\OPT$ via binary search. We then define $Q$ as a partition matroid constructed as follows: for each good, create $|\N|$ copies, one corresponding to each agent. All copies of a given good belong to the same part, and there are $|\M|$ such parts, corresponding to each good. An independent set of this matroid can include at most one element from each part. This construction naturally defines an allocation: each selected copy corresponds to a specific agent, and that agent receives the associated good. The agents’ valuation functions extend naturally while preserving submodularity: given an independent set of the matroid, the value $F_i$ for agent $i$ is computed by considering only the copies associated with that agent. This construction first appeared in \cite{lehmann2001combinatorial}. Assuming that $\OPT$ is a valid estimate of a value that each agent can achieve, we use Theorem \ref{thm:cont-greedy} to find an $\vecx \in Q$ that achieves $V_i(\bx_i) \geq (1-1/e-\eps) OPT$ for each agent. To round this fractional solution, we apply the \textsf{NonUniformPipageRounding} to $\vecx$, and obtain an integral allocation $\A = (A_1, \ldots, A_{|\N|})$ such that $v_i(A_i) \geq V_i(\vecx_i) - v_i(\ell_i)$ where $\ell_i$ is the largest valued good for agent $i$. Therefore, we obtain $v_i(A_i) \geq (1-1/e-\epsilon)\OPT - \text{Max}$.
\end{proof}
As a corollary, we obtain the following result when $\text{Max} \leq \epsilon \OPT$.
\begin{corollary}
  \label{cor:max-min-small-items}
  Consider an instance of Submodular Santa Claus with a promise that there is a feasible allocation $(S_1^*,S_2^*,\ldots,S_{\N}^*)$ of value $\OPT$ such that $v_i(j) \le \eps \OPT$ for each $i \in \N$ and each $j \in S_i$.
  In this setting, for any fixed $\eps' > 0$, there is an efficient randomized algorithm that with high probability either outputs an allocation $(S_1,S_2,\ldots,S_{|\N|})$ such that $v_i(S_i) \ge (1-1/e - \eps-\eps')\OPT$ or correctly refutes the promise.  \end{corollary}

\subsubsection{Constant Number of Agents}
Here we consider the setting with a fixed number of agents. In the additive setting there is a $\sf {PTAS}$ for this case. For the submodular setting the following result can be obtained. 

\begin{theorem}\label{thm:const-sc}
 For any fixed $\eps > 0$ and any fixed number of agents $|\N| \geq 1$, there is a $(1-1/e-\eps)$-approximation for the Submodular Santa Claus problem.
\end{theorem}
\begin{proof}[Proof Sketch.]
    
We sketch the ideas since the essence of this has already been captured in Theorem~\ref{thm:add-santa-claus-result}. Suppose we have guessed $\OPT$. Since $|\N|$ is fixed we can guess all the ``large'' items for each agent $i$ from some fixed optimum allocation $(S_1^*,\ldots,S_{|\N|}^*)$. More formally, for each agent $i$ we guess a set $A_i \subset S_i^*$ such that for all $j \in S_i^* \setminus A_i$, $v_i(A_i + j) - v_i(A_i) \le \eps \OPT$. It is easy to see that there exists such a set with $|A_i| \le (1/\eps + 1)$. Thus there are ${|\M|}^{O({|\N|}/\eps)}$ guesses that the algorithm needs to try and this is polynomial when $|\N|$ is fixed. Once these are guessed we can then find a fractional allocation and round as in the proof of Theorem~\ref{thm:add-santa-claus-result}. 
The one good loss for an agent $i$ in this setting is at most $\epsilon v_i(A_i)$ due to the preprocessing. The running time is dominated by the guessing step which incurs an overhead of ${|\M|}^{O({|\N|}/\eps)}$.
\end{proof}
\begin{remark}\label{rmk:max-min}
    
One can use a different rounding approach for the second step after the guessing the large items and solving the relaxation. This was the approach that was outlined in \cite{chekuri2010dependent} and is based on randomly allocating the items according the fractional allocation; more precisely we allocate each item $j$ independently to exactly one agent where the probablity of agent $i$ receiving item $j$ is precisely $x_{i,j}$.  Since the fractional items are ``small'', one can then use concentration properties of submodular functions, to argue that each agent will have their value preserved to within a $(1-\eps)$-factor of their fractional value with sufficiently high probability; one then uses union bound over all agents. However, randomized rounding incurs larger errors than our rounding method. To make randomized rounding work, we need to use the fact that the number of agents is small, and we need $\epsilon = \Omega(1/\log |\N|)$ for the union bound rather than a fixed constant as above. 
\end{remark}
\subsection{Nash social welfare with submodular valuations}\label{sec:nsw}

Nash Social Welfare ({\sf NSW}) is another possible objective in allocation problems: Given a set $\M$ of $m$ items, and a set $\N$ of $n$ agents with valuations $(v_i: i \in \N)$, find an allocation $\A = (A_1, \ldots, A_n)$ to maximize the Nash social welfare objective, 
$$ \left(\prod_{i \in \N} v_i(A_i)\right)^{\sfrac{1}{n}}.$$
In this section, we present a new $\frac{1}{5}$-approximation algorithm for this problem with monotone submodular valuations, using the rounding lemma. 

As mentioned in Section \ref{sec:nsw-intro}, our algorithm follows the matching-rematching approach that has been used for maximizing \NSW with submodular valuations \cite{GKKsubmodnsw, li2022constant,garg2023approximating}. \cite{li2022constant} used this framework with a log-multilinear relaxation and randomized rounding, which gave the first constant factor for \NSW  with submodular valuation \cite{li2022constant}; however, the factor was rather small ($\approx 1/380$).
Later, the factor was significantly improved to $1/4$ \cite{garg2023approximating}, by replacing the log-multilinear relaxation by discrete local search. Here we return to the log-multilinear relaxation and show that it can be actually used to extract a good approximation factor as well. This is accomplished primarily by the new rounding lemma. In addition, we provide an alternative way to solve the log-multilinear relaxation (approximately) by continuous local search; this also helps in improving the approximation factor, but the effect of this improvement is relatively minor (essentially replacing $1/e$ by $1/2$).

\subsubsection{A $\frac{1}{5}$-Approximation Algorithm}

Let us describe our new algorithm. It consists of 3 phases: (1) initial matching, (2) relaxation, (3) rematching. The relaxation phase is the new contribution here: we solve the log-multilinear relaxation using continuous local search, and then round it using the non-uniform rounding lemma. A formal description of the algorithm follows.

\begin{algorithm}[h!]
    \caption{$1/5$-approximation for $\NSW$ with submodular valuations}
    \label{alg:submod-nsw}
    \SetKwInOut{Input}{Input}\SetKwInOut{Output}{Output}
  \Input{Instance $\MMSinssym$}
  \Output{An allocation of items with $\NSW$ at least $\frac{1}{5} \OPT$}
  \BlankLine
  $\tau :=$ matching $\N \to \M$ maximizing $\sum_{i \in \N} \log v_i(\tau(i)$ \\
  $H := \tau[\N], \M' = \M \setminus H$ \\
  $\bx := \mbox{\sf ContinuousLocalSearch}(\M', \{v_i: i \in \N \}) $ \\
  $ (S_1,\ldots,S_n) := \mbox{\sf NonUniformPipageRound}(\bx, \{v_i: i \in \N\}) $ \\
  $\sigma := $ matching $\N \to \M$ maximizing $\sum_{i \in \N} \log v_i(S_i+\sigma(i))$ \\
  \Return $(S_i+\sigma(i): i \in \N)$
\end{algorithm}

\begin{algorithm}[h!]
    \caption{{\sf ContinuousLocalSearch}$(\M', (v_i: i \in \N)$}
    \label{alg:continuous-local-search}
    \SetKwInOut{Input}{Input}\SetKwInOut{Output}{Output}
   Let $P := \{ \bx \in \R_+^{\M' \times \N}: \forall j \in \M, \sum_{i \in \N} x_{ij} = 1 \}$ (the assignment polytope) \\
  Initialize $\vecx^{(0)} := \frac{1}{n} \bone \in P$ , and $t := 0$\\
    Define the function $F(\vecx) = \sum_{i \in [n]}\log V_i(\vecx_i)$.\\
  \While{$\exists \by \in P, (\by-\bx) \cdot \nabla F(\bx) > 0$} {
    $\vecx^{(t+1)} := \vecx^{(t)} + \epsilon (\by - \vecx^{(t)})$ \\
    $t=t+1$  }
\end{algorithm}

First, we have the following simple lemma about the initial matching.

\begin{lemma}
If $\tau:\N \to \M$ is the initial optimal matching, then for every agent $i \in \N$ and every remaining item $j \in \M \setminus \tau[N]$, $v_i(j) \leq v_i(\tau(i))$.
\end{lemma}

\begin{proof}
If $v_i(j) > v_i(\tau(i))$, then we could improve the matching by redefining $\tau(i) := j$; note that this item is not allocated to anybody else in the matching. 
\end{proof}

Next, we analyze the continuous local search algorithm.
Here we ignore numerical issues of convergence of the algorithm; we assume that the output is a local optimum, which can be implemented up to an arbitrarily small error.
Let $P$ denote the assignment polytope as in the description of {\sf ContinuousLocalSearch}. 

\begin{lemma}
\label{lem:local-opt}
Suppose that $\bx \in P$ is the output of {\sf ContinuousLocalSearch}, and $\bx^* \in P$ is any other feasible solution. Then
$$ \sum_{i \in \N} \frac{V_i(\bx^*_i)}{V_i(\bx_i)} \leq 2|\N|.$$
\end{lemma}

\begin{proof}
If $\bx$ is a local optimum, then we have
$$ (\bx^* - \bx) \cdot \nabla F(\bx) = \sum_{i \in \N} (\bx^*_i - \bx_i) \cdot \nabla (\log V_i(\bx_i)) \leq 0.$$
Consider the points $\bx \vee \bx^*$ and $\bx \wedge \bx^*$ (the coordinate maximum and minimum). We can write $(\bx \vee \bx^*) + (\bx \wedge \bx^*) = \bx + \bx^*$, and hence
$$ \sum_{i \in \N} ((\bx_i \vee \bx_i^*) - \bx_i) \cdot \nabla (\log V_i(\bx_i)) \leq \sum_{i \in \N} (\bx_i - (\bx_i \wedge \bx_i^*)) \cdot \nabla (\log V_i(\bx_i)).$$
By the chain rule, $\nabla (\log V_i(\bx_i)) = \frac{\nabla V_i(\bx_i)}{V_i(\bx_i)}$.
Also, by submodularity of $V_i$, we have $((\bx_i \vee \bx_i^*) - \bx_i) \cdot \nabla V_i(\bx_i) \geq V_i(\bx_i \vee \bx^*_i) - V_i(\bx_i)$. Similarly, $(\bx_i - (\bx_i \wedge \bx_i^*)) \cdot \nabla V_i(\bx_i) \leq V_i(\bx_i) - V_i(\bx_i \wedge \bx^*_i)$
From here and the inequality above,
$$ \sum_{i \in \N} \frac{V_i(\bx_i \vee \bx^*_i) - V_i(\bx_i)}{V_i(\bx_i)} \leq
\sum_{i \in \N} \frac{V_i(\bx_i) - V_i(\bx_i \wedge \bx_i^*)}{V_i(\bx_i)}.$$
This inequality is valid even for non-monotone submodular functions. Since here we are dealing with monotone submodular functions, we can drop $V_i(\bx_i \wedge \bx_i^*)$ and estimate $V_i(\bx_i \vee \bx_i^*) \geq V_i(\bx_i^*)$, which gives
$$ \frac{1}{n} \sum_{i \in \N} \frac{V_i(\bx_i^*)}{V_i(\bx_i)} \leq 2.$$
\end{proof}

Next, we apply the non-uniform pipage rounding procedure. Given a fractional solution $\bx$, it produces an allocation $(S_1,\ldots,S_n)$ of the items in $\M'$ such that
$$ v_i(S_i) \geq V_i(\bx_i) - v_i(\ell(i)) $$
where
$$ \ell(i) = \argmax_{j \in \M'} v_i(j).$$

For the analysis of the last stage, we need the following ``rematching lemma'', which has appeared in various forms in this line of work. The variant we use here appears in the latest arxiv version of \cite{garg2023approximating}.

\begin{lemma}
\label{lem:rematching}
If $\tau:\N \to \M$ is an optimal initial matching, $H = \tau[\N]$, $\pi:\N \to H$ is any other matching on $H$, $(S_i: i \in \N)$ is any partition of $\M' = \M \setminus H$, and $\ell(i) = \argmax_{j \in \M'} v_i(j)$, then there is a matching $\sigma:\N \to H$ such that 
$$ \prod_{i \in \N} v_i(S_i + \sigma(i)) \geq \prod_{i \in \N} \max \{ v_i(S_i), v_i(\ell(i)), v_i(\pi(i)) \}.$$
\end{lemma}

Using the rematching lemma, we conclude the analysis as follows. 

\begin{theorem}
If $\OPT$ is the \NSW value of the optimal solution, then our algorithm finds a solution of value
$$ \left(\prod_{i \in \N} v_i(S_i + \sigma(i)) \right)^{1/n} \geq \frac{1}{5} \OPT.$$
\end{theorem}

\begin{proof}
Let the optimal allocation be $(H_i \cup O_i: i \in \N)$, where $H_i \subset H$ and $O_i \subset \M' = \M \setminus H$.
Let $\pi: \N \to H$ be a matching such that $\pi(i) = \argmax_{j \in H_i} v_i(j)$ if $H_i \neq \emptyset$, and otherwise $pi(i)$ an arbitrary item in $H$ so that $\pi$ is matching. By submodularity, we have
$$ v_i(H_i \cup O_i) \leq v_i(O_i) + |H_i| v_i(\pi(i)) = V_i(\bx^*_i) + |H_i| v_i(\pi(i)) $$
where $\bx^*_i = \bone_{O_i}$.

From Lemma~\ref{lem:rematching}, we have that the value obtained by our algorithm is
$$ \ALG = \left(\prod_{i \in \N} v_i(S_i + \sigma(i)) \right)^{1/n} \geq \left( \prod_{i \in \N} W_i \right)^{1/n} $$
where $W_i = \max \{ v_i(S_i), v_i(\ell(i)), v_i(\pi(i)) \}$. From the pipage rounding procedure, we get $v_i(S_i) \geq V_i(\bx_i) - v_i(\ell(i))$. Hence,
$$W_i \geq \max \{ V_i(\bx_i) - v_i(\ell(i)), v_i(\ell(i)), v_i(\pi(i) \}
\geq \max \{ \frac12 V_i(\bx_i), v_i(\pi(i)) \}.$$
Finally, we estimate by AM-GM
\begin{eqnarray*}
\frac{\OPT}{\ALG} &\leq & \left( \prod_{i \in \N} \frac{v_i(H_i \cup O_i)}{W_i}   \right)^{1/n}  \leq \frac{1}{n} \sum_{i \in \N} \frac{v_i(H_i \cup O_i)}{W_i}    \\
& \leq & \frac{1}{n} \sum_{i \in \N} \frac{V_i(\bx^*_i) + |H_i| v_i(\pi(i))}{W_i}
\leq \frac{1}{n} \sum_{i \in \N} \left( \frac{V_i(\bx_i^*)}{\frac12 V_i(\bx_i)} + |H_i| \right).
\end{eqnarray*}
Using Lemma~\ref{lem:local-opt} and the fact that $\sum_{i \in \N} |H_i| \leq n$, we conclude that $\frac{\OPT}{\ALG} \leq 5$.
\end{proof}

\subsubsection{Constant number of agents}
We can prove the following theorem for $\NSW$ with a constant number of agents This theorem also works for weighted $\NSW$. We note that a similar result was known previously  \cite{GKKsubmodnsw}. See Remark \ref{rmk:nsw} for a comparison.
\begin{theorem}\label{thm:cycle-nsw}
    Given an instance of the fair allocation problem, $\MMSinssym$, we can compute in randomized polynomial time an allocation $A_i \in \Pi_n([m])$ such that the value of $\NSW$ of the allocation is $(1-1/e-\epsilon) \OPT$ where $\OPT$ is the value of the optimal Nash social welfare.
\end{theorem}
\begin{proof}[Proof Sketch.]
    The proof is similar to proof of constantly many agents for the max-min problem (see \ref{thm:const-sc}) and we only sketch the idea here. First, we predict the value agent $i$ receives in an optimal allocation, say $\OPT_i$. For a particular agent, we can guess this to a $(1+\epsilon)$ multiplicative approximation using a grid search in polynomial time. Therefore, for constantly many agents, we can guess the valuation profile for an optimal $\NSW$ solution in polynomial time. Now, consider the corresponding optimal allocation, $(S_1^*, \ldots, S_n^*)$. From each $S_i^*$, we guess a subset $A_i \subseteq S_i^*$ such that for each $g \in S_i^* \setminus A_i$, $v_i(g \mid A_i) \leq \epsilon \OPT_i$. It is clear that we can have an $A_i$ such that $|A_i| \leq (1/\epsilon)$ otherwise, we have a value of more than $\OPT_i$ for agent $i$. Therefore, together for all $n$ agents, this $A_i$ can be guessed in $m^{O(\sfrac{n}{\epsilon})}$ time. This is clearly polynomial when $n$ and $\epsilon$ are constants. Now, we find a fractional solution for the $\NSW$ problem -- for each agent, we define it's new valuation function to be $v_{i \mid A_i}$ which is the function that gives marginal values of sets on $A_i$. It is well-known and easy to see that this function is also submodular. For this function, we aim to allocate agent $i$ value $\OPT_i - v_i(A_i)$. Again, following the Theorem \ref{thm:cont-greedy} from previous section, if such an allocation exists, we can compute a fractional allocation $\vecx$ in time polynomial in $m$, $n$ and $\epsilon'$ such that $v_i(\vecx_i) \geq (1-1/e-\epsilon') (\OPT_i - v_i(A_i))$. Further, we can also ensure that no agent is allocated a good of value more than $\epsilon \cdot \OPT$ by making the value of any such good for the agent zero (this also maintains submodularity). Then we run the non-uniform pipage rounding algorithm and obtain an allocation that allocates each agent a value of $(1-1/e - \epsilon'-\epsilon)\OPT_i$. Therefore, in time polynomial in $m$, $n$ and $\epsilon'' = \epsilon' + \epsilon$, we can compute an allocation that maximizes $\NSW$ with constant number of agent up to a factor of $(1-1/e-\epsilon'')$ i.e., we have a PTAS for this problem with constantly many agents. Note that this algorithm also works with asymmetric agents. 
\end{proof}
\begin{remark}\label{rmk:nsw}
    A $(1-1/e-\eps)$-approximation for constant $n$ was achieved previously \cite{GKKsubmodnsw}. That work relied on an integral version of Theorem \ref{thm:cont-greedy} which in turn was based on the same ideas as in max-min allocation from \cite{chekuri2010dependent}.  Therefore, to obtain an approximation guarantee of $(1-1/e-\epsilon)$ for any fixed $\epsilon > 0$, \cite{GKKsubmodnsw} requires guessing from an optimal allocation all goods allocated to an agent $i$ of value at least $\epsilon \cdot \OPT_i/\log n$. In the preceding theorem, for a fixed $\epsilon > 0$, we need to guess the ``large'' goods of value at least $\epsilon \OPT_i$. Therefore the size of the set to be guessed is reduced.
\end{remark}
\subsection{Approximating $\MMS$ for Submodular Valuations}\label{sec:mms}
$\MMS$ is another fairness objective frequently studied in allocation problems. It is a share-based guarantee, i.e., each agent demands a certain value a.k.a. share. An allocation is considered to be fair if all agents receive their share. In case of $\MMS$, this value for a particular agent $i$ depends only on the valuation function of the agent, $v_i$ and the number of agents $|\N|$. It is defined as follows:
\[
\MMS^{|\N|}_i(\M) \coloneqq \max_{A \in \Pi_{|\N|}{(\M)}} \min_{i \in [n]} v_i(A_i).
\]
Essentially, the $\MMS$ value of an agent is the maximum value she can ensure for herself when she divides the goods into $|\N|$ parts and chooses the worst part for herself. When $\M$ and $\N$ are clear from the context, we denote $\MMS^{|\N|}_i(\M)$ as simply $\MMS_i$. As mentioned in Section \ref{sec:mms-intro}, since $\MMS$-allocations do not exist, we study $\alpha$-$\MMS$ allocations for $\alpha (\leq 1)$ as large as possible. Here, we give an algorithm that achieves a value of $\approx 1/2(1-1/e) \approx 0.316$ based on a simple application of the $\textsf{NonUniformPipageRounding}$. Our approach is again to round a fractional allocation. While in the previous two sections we solved some fractional program to find the fractional allocation, in this section, we give a direct solution using two known results. To state the results, we need to define another extension of submodular valuation functions called the concave extension. Given a submodular function, $f$ the concave extension, $f^+$ of $f$ is defined as follows.
\begin{definition}[Concave extension of $f$]\label{def:concave-ext}
\begin{equation*}
f^+(\vecx) = \max_{(\alpha_S)_{S \in 2^{\M}}} \{ \sum_{S} f(S) \alpha_S: \sum_{S} \alpha_S = 1 \text{~and~} \sum_{S \ni j} \alpha_S = x_j \quad \forall j \in \M \text{~and~} \alpha_S \ge 0 \quad \forall S\in 2^{\M}\}.
\end{equation*}
\end{definition}
Now, we state the aforementioned results. The first one, from \cite{chekuri20241} gives an upper bound on $\MMS$ in terms of the concave extension. The second one, from \cite{calinescu2007maximizing}, is called the correlation gap and connects the value of the concave extension at any fractional point with value of the multilinear extension at the same point.

\begin{lemma}[\cite{chekuri20241}]
    Given an instance of the fair division problem, $\MMSinssym$ with submodular valuations, $\MMS_i \leq v_i^+(\frac{1}{|\N|}(\bone^{|\M|}))$ for all $i \in \N$, where $v_i^+$ is the concave extension of the function $v_i$.
\end{lemma}
\begin{claim}[\cite{calinescu2007maximizing}]\label{clm:correlation-gap}
    For any monotone submodular function, $f$ with concave extension $f^+$ and multilinear extension $F$,  \[F(\vecx) \geq \left(1-\frac{1}{e}\right)f^+(\vecx).\]
\end{claim}
Combining the above two results, we get,
\begin{claim}\label{clm:submod-mms-ub}
    Given an instance of the fair division problem, $\MMSinssym$ with submodular valuations, $\MMS_i \leq \left( \frac{e}{e-1} \right) \cdot V_i(\frac{1}{|\N|}(\bone^{|\M|}))$.
\end{claim}
We can now use $\textsf{NonUniformPipageRounding}$ to obtain: 
\begin{theorem}\label{thm:mms-generic}
    Given a fair division instance, $\MMSinssym$ with submodular valuations, if $\max_{i, j} v_i(\{j\}) \leq \epsilon \cdot \MMS_i$ for all $i \in \N, j \in \M$,  we can compute an allocation in randomized polynomial time where each agent receives a value of $\left(1-\frac{1}{e} - \epsilon\right) \MMS_i$.
\end{theorem}
\begin{proof}
    We run $\textsf{NonUniformPipageRounding}$ on the allocation $\vecx=(\frac{1}{|\N|}\bone^{|\M|})$. Therefore, we get an integral allocation, $\A = (A_1, \ldots, A_{|\N|})$ such that $v_i(A_i) \geq V_i(\frac{1}{|\N|}(\bone^{|\M|})) - v_i(\ell_i)$ where $\ell_i \coloneqq \argmax_{j}v_i(\{j\})$. Since by assumption $v_i(\ell_i) \leq \epsilon \MMS_i$ by assumption, and from Claims \ref{clm:correlation-gap} and \ref{clm:submod-mms-ub}, $V_i(\frac{1}{|\N|}(\bone^{|\M|})) \geq \left(1-\frac{1}{e}\right)\MMS_i$, the theorem follows.
\end{proof}

\begin{algorithm}[h!]
    \caption{$\frac{1}{2}(1-\frac{1}{e}-o(1))$-approximation for $\MMS$ with submodular valuations}
    \label{alg:submod-mms}
    \SetKwInOut{Input}{Input}\SetKwInOut{Output}{Output}
  \Input{Instance $\MMSinssym$}
  \Output{An allocation where each agent receives at least $\frac{1}{2}(1-\frac{1}{e}-o(1))\MMS_i$.}
  \BlankLine
  Initialize $\N' \leftarrow \N$ and $\M' \leftarrow \M$\\
  \While{there exists $i \in \N'$ and $g \in \M'$ with $v_i(g) \geq \frac{1}{2}V_i(\frac{1}{|\N'|}\bone^{|\M'|})$ \label{step:while-start}}{Allocate $g$ to $i$ \\ Update $\N' \leftarrow \N' \setminus i$ and $\M' \leftarrow \M' \setminus g$\label{step:while-end}}
  Initialize $\vecx = \left(\frac{1}{|\N'|} (\bone^{|\N'| \times |\M'|})\right)$\\
  $(A_1, \ldots, A_{|\N'|}) \leftarrow \textsf{NonUniformPipageRound}(\vecx)$ \label{step:pipage-round} \\
  \Return $(A_1, \ldots, A_{|\N|})$
\end{algorithm}

Algorithm \ref{alg:submod-mms} gives a $\approx 0.316$-approximation for $\MMS$ with submodular valuations. To analyze this, we first understand the algorithm. Step (\ref{step:pipage-round}) is performing our $\sf NonUniformPipageRounding$. The while loop in lines (\ref{step:while-start}) to (\ref{step:while-end}) performs an operation that is called as ``single good reduction''. The following claim regarding this is well-known and used in almost all works on $\MMS$. 
\begin{claim}\label{clm:single-good-redn}
    Given a fair division instance $\MMSinssym$, the $\MMS$ value of any agent is retained if we remove any single agent and any single good. That is,
    \begin{equation*}
        \MMS_i^{|\N|}(\M) \leq \MMS_i^{|\N|-1}(\M \setminus \{g\}) \qquad \text{ for all } g \in \M
    \end{equation*}
\end{claim}
We can now prove the following result for Algorithm \ref{alg:submod-mms}.
\begin{theorem}
    Given a fair division instance, $\MMSinssym$ with submodular valuations, Algorithm \ref{alg:submod-mms} outputs an allocation where each agent receives a value of $\frac{1}{2}\left(1-\frac{1}{e} - o(1)\right)\MMS_i$.
\end{theorem}
\begin{proof} 
First, let us assume that we can compute $V_i(\by)$ exactly for any $\by \in [0,1]^{|\N| \times |\M|}$. Note that as per Algorithm \ref{alg:submod-mms}, an agent can receive an allocation either in the while loop of Steps \ref{step:while-start}-\ref{step:while-end} or in Step \ref{step:pipage-round} via \textsf{NonUniformPipageRounding}. In every iteration of the while loop, we remove one agent and one good. Therefore, by Claim \ref{clm:single-good-redn}, the $\MMS$ value of any agent in the residual instance does not decrease. As a result, until the algorithm reaches Step \ref{step:pipage-round}, we always have $\MMS_i \leq (\frac{e}{e-1})V_i(\frac{1}{|\N'|}\bone^{|\M'|})$. Therefore, any agent who gets a good in the while loop receives a good of value at least $\frac{1}{2}(1-\frac{1}{e})\MMS_i$. Now, in Step \ref{step:pipage-round}, by Lemma \ref{lem:nonunif-pipage}, we get an allocation of value at least $V_i(\frac{1}{|\N'|}\bone^{|\M'|}) - v_i(\ell_i) \geq \frac{1}{2}V_i(\frac{1}{|\N'|}\bone^{|\M'|}) \geq \frac{1}{2}(1-\frac{1}{e})\MMS_i$. 

Now, we note that the values of $V_i(\frac{1}{|\N'|}\bone^{|\M'|})$ are computed by sampling and in polynomial time, we can only compute them to a factor of $(1-\epsilon)$. Accounting for this, we get that in randomized polynomial time we can guarantee each agent a value of $(1-\frac{1}{e}-o(1))\MMS_i$. This proves the theorem.
\end{proof}

\paragraph{Acknowledgement.}
We would like to thank L\'{a}szl\'{o} V\'{e}gh for useful discussions, in particular posing a question about using continuous local search for solving the multilinear relaxation of the Nash social welfare problem.

\printbibliography
\end{document}